\documentclass{article}
\usepackage{times,authblk}
\usepackage[UKenglish]{babel}
\usepackage{amsmath, amsfonts, amssymb, amsthm, bm, stmaryrd, graphicx}
\renewcommand{\epsilon}{\varepsilon}
\renewcommand{\phi}{\varphi}
\renewcommand{\rho}{\varrho}
\usepackage{upgreek}
\renewcommand{\tau}{\uptau}

\newtheorem{Def}{Definition}[section]
\newenvironment{definition}{\begin{Def} \rm}{\end{Def}}
\newtheorem{lemma}[Def]{Lemma}
\newtheorem{proposition}[Def]{Proposition}

\newtheorem{theorem}[Def]{Theorem}
\newtheorem{example}[Def]{Example}
\newtheorem{remark}[Def]{Remark}

\newcommand{\komma}{,\hspace{0.6em}}
\renewcommand{\leq}{\leqslant}
\renewcommand{\geq}{\geqslant}
\renewcommand{\emptyset}{\varnothing}

\newcommand{\Naturals}{{\mathbb N}}
\newcommand{\Integers}{{\mathbb Z}}

\newcommand{\Reals}{{\mathbb R}}

\newcommand{\abs}[1]{{\left| #1 \right|}}

\newcommand{\CircleGroup}{{\mathbf C}} 
\newcommand{\Aut}[1]{{\operatorname{Aut}(#1)}}
\newcommand{\Circ}[1]{\operatorname{Circ}(#1)}
\newcommand{\U}{\text{\rm U}}
\renewcommand{\O}{\text{\rm O}}
\newcommand{\SO}{\text{\rm SO}}
\newcommand{\PSO}{\text{\rm PSO}}

\newcommand{\cc}{^{\bot\bot}}
\renewcommand{\c}{^\bot}
\newcommand{\indistinguishable}{\mathbin{\equiv}}
\newcommand{\herm}[2]{(#1,#2)}
\newcommand{\lin}[1]{[#1]}
\newcommand{\withoutzero}{^{\raisebox{0.2ex}{\scalebox{0.4}{$\bullet$}}}}
\newenvironment{smm}{\begin{pmatrix}}{\end{pmatrix}}
\newcommand{\similar}{\mathbin{\approx}}
\newcommand{\congr}{\mathbin{\theta}}

\begin{document}

\title{Gradual transitivity \\ in orthogonality spaces of finite rank}

\author{Thomas Vetterlein}

\affil{\footnotesize Department of Knowledge-Based Mathematical Systems \\
Johannes Kepler University Linz \\
Altenberger Stra\ss{}e 69, 4040 Linz, Austria \\
{\tt Thomas.Vetterlein@jku.at}}

\date{\today}

\maketitle

\begin{abstract}\parindent0pt\parskip1ex

An orthogonality space is a set together with a symmetric and irreflexive binary relation. Any linear space equipped with a reflexive and anisotropic inner product provides an example: the set of one-dimensional subspaces together with the usual orthogonality relation is an orthogonality space. We present simple conditions to characterise the orthogonality spaces that arise in this way from finite-dimensional Hermitian spaces.

Moreover, we investigate the consequences of the hypothesis that an orthogonality space allows gradual transitions between any pair of its elements. More precisely, given elements $e$ and $f$, we require a homomorphism from a divisible subgroup of the circle group to the automorphism group of the orthogonality space to exist such that one of the automorphisms maps $e$ to $f$, and any of the automorphisms leaves the elements orthogonal to $e$ and $f$ fixed. We show that our hypothesis leads us to positive definite quadratic spaces. By adding a certain simplicity condition, we furthermore find that the field of scalars is Archimedean and hence a subfield of the reals.

\mbox{}\vspace{-2ex}

\end{abstract}

\vspace{2ex}

\section{Introduction}
\label{sec:introduction}

An orthogonality space is a set endowed with a binary relation that is supposed to be symmetric and irreflexive. The notion was proposed in the 1960s by David Foulis and his collaborators \cite{Dac,Wlc}. Their motivation may be seen as part of the efforts to characterise the basic model used in quantum physics: the Hilbert space. The strategy consists in reducing the structure of this model to the necessary minimum. Compared to numerous further approaches that have been proposed with a similar motivation \cite{EGL1, EGL2}, we may say that Foulis's concept tries to exhaust the limits of abstraction, focussing solely on the relation of orthogonality. The prototypical example of an orthogonality space is the projective Hilbert space together with usual orthogonality relation. Just one aspect of physical modelling is this way taken into account -- the distinguishability of observation results.

We have dealt with the problem of characterising the complex Hilbert spaces as orthogonality spaces in our recent work \cite{Vet1,Vet2}. The idea was to make hypotheses on the existence of certain symmetries. In the infinite-dimensional case, just a few simple assumptions led to success \cite{Vet2}, whereas in the finite-dimensional case, the procedure was considerably more involved \cite{Vet1}.

In the present paper, we first of all point out a straightforward way of limiting the discussion to inner-product spaces. We deal here with the finite-dimensional case, that is, we assume all orthogonality spaces to have a finite rank. We introduce the notion of linearity and establish that any linear orthogonality space of a finite rank $\geq 4$ arises from an (anisotropic) Hermitian space over some skew field.

On this basis, we are furthermore interested in finding conditions implying that the skew field is among the classical ones. However, to determine within our framework the characteristic properties of, say, the field of complex numbers is difficult and we are easily led to the choice of technical, physically poorly motivated hypotheses. Rather than tailoring conditions to the aim of characterising a particular field of scalars, we focus in this work on an aspect whose physical significance is not questionable: we elaborate on the principle of smooth transitions between states. A postulate referring to this aspect might actually be typical for any approach to interpret the quantum physical formalism; cf., e.g., \cite{Har}. Our condition looks as follows. Let $e$ and $f$ be distinct elements of an irredundant orthogonality space. Then we suppose that an injective homomorphism from a subgroup of the abelian group of unit complex numbers to the group of automorphisms exists, the action being transitive on the closure of $e$ and $f$ and fixing elements orthogonal to $e$ and~$f$.

The complex Hilbert space does not give rise to an example of the orthogonality spaces considered here, but the real Hilbert space does. The natural means of visualising matters is an $n$-sphere, which nicely reflects the possibility of getting continuously from any point to any other one by means of a rotation, in a way that anything orthogonal to both is left at rest. As the main result of this contribution, we establish that any linear orthogonality space of finite rank that fulfils the afore-mentioned hypothesis regarding the existence of automorphisms arises from a positive definite quadratic space. We furthermore subject the orthogonality space to a simplicity condition, according to which there are no non-trivial quotients compatible with the automorphisms in question. We show that the field of scalars is then embeddable into the reals.

The paper is organised as follows. In Section~\ref{sec:preliminaries}, we recall the basic notions used in this work and we compile some basic facts on inner-product spaces and the orthogonality spaces arising from them. In Section~\ref{sec:Hermitian-spaces}, we introduce linear orthogonality spaces; we show that the two simple defining conditions imply that an orthogonality space arises from a Hermitian space over some skew field. In Section~\ref{sec:ordered-fields}, we formulate the central hypothesis with which we are concerned in this paper, the condition that expresses, in the sense outlined above, the gradual transitivity of the space. We show that, as a consequence, the skew field is commutative, its involution is the identity, and it admits an order. The subsequent Section~\ref{sec:circulation-group} is devoted to the group generated by those automorphisms that occur in our main postulate. In Section~\ref{sec:spaces-over-subfields-of-R}, we finally show that the exclusion of certain quotients of the orthogonal space implies that the ordered field actually embeds into $\Reals$. An outlook on possible continuations of this work can be found in the concluding Section~\ref{sec:conclusion}.

\section{Orthogonality spaces}
\label{sec:preliminaries}

We investigate in this paper relational structures of the following kind.

\begin{definition} \label{def:OS}
An {\it orthogonality space} is a non-empty set $X$ equipped with a symmetric, irreflexive binary relation $\perp$, called the {\it orthogonality relation}.

We call $n \in \Naturals$ the {\it rank} of $(X, \perp)$ if $X$ contains $n$ but not $n + 1$ mutually orthogonal elements. If $X$ contains $n$ mutually orthogonal elements for any $n \in \Naturals$, then we say that $X$ has infinite rank.
\end{definition}

This definition was proposed by David Foulis; see, e.g., \cite{Dac,Wlc}. The idea of an abstract orthogonality relation has been taken up by several further authors \cite{Mac,Fin,Pul,HePu,Rod,Bru}, although definitions sometimes differ from the one we use here. It should be noted that the notion of an orthogonality space is very general; in fact, orthogonality spaces are essentially the same as undirected graphs.

Orthogonality space naturally arise from inner-product spaces. We shall compile the necessary background material; for further information, we may refer, e.g., to \cite{Gro,Piz,Sch}.

By a {\it $\star$-sfield}, we mean a skew field (division ring) $K$ together with an involutorial antiautomorphism $^\star \colon K \to K$. We denote the centre of $K$ by $Z(K)$ and we let $U(K) = \{ \epsilon \in K \colon \epsilon \epsilon^\star = 1 \}$ be the set of unit elements of $K$.

Let $H$ be a (left) linear space over the $\star$-sfield $K$. Then a {\it Hermitian form} on $H$ is a map $\herm{\cdot}{\cdot} \colon H \times H \to K$ such that, for any $u, v, w \in H$ and $\alpha, \beta \in K$, we have
\begin{align*}
& \herm{\alpha u + \beta v}{w} \;=\; \alpha \, \herm{u}{w} + \beta \, \herm{v}{w}, \\
& \herm{w}{\alpha u + \beta v} \;=\;
                               \herm{w}{u} \, \alpha^\star + \herm{w}{v} \, \beta^\star, \\
& \herm{u}{v} \;=\; \herm{v}{u}^\star.
\end{align*}
The form is called {\it anisotropic} if $\herm{u}{u} = 0$ holds only if $u = 0$.

By a {\it Hermitian space}, we mean a linear space $H$ endowed with an anisotropic Hermitian form. If the $\star$-sfield $K$ is commutative and the involution $\star$ is the identity, then we refer to $H$ as a {\it quadratic space}. We moreover recall that a field $K$ is {\it ordered} if $K$ is equipped with a linear order such that (i)~$\alpha \leq \beta$ implies $\alpha + \gamma \leq \beta + \gamma$ and (ii)~$\alpha, \beta \geq 0$ implies $\alpha \beta \geq 0$. If $K$ can be made into an ordered field, $K$ is called {\it formally real}. If $K$ is an ordered field and we have that $\herm u u > 0$ for any $u \in H \setminus \{0\}$, then $H$ is called {\it positive definite}.

As usual, we write $u \perp v$ for $\herm u v = 0$, where $u, v \in H$. Applied to subsets of $H$, the relation $\perp$ is understood to hold elementwise. Moreover, we write $\lin{u_1, \ldots, u_k}$ for the subspace spanned by non-zero vectors $u_1, \ldots, u_k \in H$. For a subspace $E$ of $H$, we write $ E\withoutzero = E \setminus \{0\}$ and we define $P(E) = \{ \lin u \colon u \in E\withoutzero \}$. That is, $P(H)$ is the (base set of the) projective space associated with $H$.

We may now indicate our primary example of orthogonality spaces.

\begin{example}
Let $H$ be a Hermitian space. Because the form is reflexive, $\lin u \perp \lin v$ is equivalent to $\lin v \perp \lin u$ for any $u, v \in H\withoutzero$, and because the form is anisotropic, $\lin u \perp \lin u$ does not hold for any $u \in H\withoutzero$. In other words, the orthogonality relation on $P(H)$ is symmetric and irreflexive and hence makes $P(H)$ into an orthogonality space.

If $H$ is finite-dimensional, the dimension of $H$ coincides with the rank of $(P(H), \perp)$. If $H$ is infinite-dimensional, $(P(H), \perp)$ has infinite rank.
\end{example}

We call an orthogonality space $X$ {\it irredundant} if, for any $e, f \in X$, $\,\{e\}\c = \{f\}\c$ implies $e = f$. For example, for any Hermitian space $H$, $(P(H), \perp)$ is irredundant. For the reasons explained in the following remark, focusing on orthogonality spaces with this property is no serious restriction.

\begin{remark} \label{rem:irredundance}
Let $(X, \perp)$ be an orthogonality space. If $X$ is not irredundant, there are distinct elements that are, by means of the orthogonality relation, indistinguishable. Roughly speaking, $X$ then arises from an irredundant space simply by multiplying some of its elements.

Indeed, for $e, f \in X$, define $e \indistinguishable f$ to hold if $\{e\}\c = \{f\}\c$. Then $\indistinguishable$ is an equivalence relation. Moreover, $e \indistinguishable e'$ and $f \indistinguishable f'$ imply that $e \perp f$ is equivalent to $e' \perp f'$. Thus the quotient set $X/{\indistinguishable}$ can be made into an orthogonality space, where we have, for any $e, f \in X$, ${e/{\indistinguishable}} \perp {f/{\indistinguishable}}$ if and only if $e \perp f$. By construction, $X/{\indistinguishable}$ is irredundant.

We conclude that, given an orthogonality space that is not irredundant, we can easily switch to an irredundant one whose structure can be considered as essentially the same.
\end{remark}

Both orthogonality spaces and Hermitian spaces can be dealt with by lattice-theoretic means.

For a subset $A$ of an orthogonality space $(X, \perp)$, we let
\[ A\c \;=\; \{ e \in X \colon e \perp A \}, \]
where it is again understood that the orthogonality relation is applied to subsets of $X$ elementwise. The map ${\mathcal P}(X) \to {\mathcal P}(X) \komma A \mapsto A\cc$ is a closure operator \cite{Ern}. If $A\cc = A$, we say that $A$ is {\it orthoclosed} and we denote the set of all orthoclosed subsets of $X$ by ${\mathcal C}(X, \perp)$. We partially order ${\mathcal C}(X, \perp)$ by set-theoretical inclusion and equip ${\mathcal C}(X, \perp)$ with the operation $\c$. In this way, we are led to an ortholattice, from which $(X, \perp)$ can in certain cases be recovered.

Following Roddy \cite{Rod}, we call an orthogonality space {\it point-closed} if, for any $e \in X$, $\{e\}$ is orthoclosed.

\begin{proposition} \label{prop:OS-ortholattice}
${\mathcal C}(X, \perp)$ is a complete ortholattice.

Moreover, $(\{ \{e\}\cc \colon e \in X \}, \;\perp)$ is an orthogonality space and the map $X \to \{ \{e\}\cc \colon e \in X \} \komma e \mapsto \{e\}\cc$ is orthogonality-preserving. If $(X,\perp)$ is point-closed, then the map $X \to (\{ \{e\} \colon e \in X \}, \perp) \komma e \mapsto \{e\}$ is an isomorphism between $(X, \perp)$ and the set of atoms of ${\mathcal C}(X, \perp)$ endowed with the inherited orthogonality relation.
\end{proposition}

\begin{proof}
The collection of closed subsets of a closure space forms a complete lattice and this fact applies to ${\mathcal C}(X, \perp)$. Moreover, $A\c$ is clearly a complement of an $A \in {\mathcal C}(X, \perp)$ and $\c \colon {\mathcal C}(X, \perp) \to {\mathcal C}(X, \perp)$ is order-reversing as well as involutive. This shows the first part.

For any $e, f \in X$, we have $\{e\}\cc \perp \{f\}\cc$ if and only if $e \perp f$. It follows that $(\{ \{e\}\cc \colon e \in X \}, \perp)$ is an orthogonality space and the assignment $e \mapsto \{e\}\cc$ is orthogonality-preserving. Moreover, if $\{e\} = \{e\}\cc$ holds for any $e \in X$, then ${\mathcal C}(X, \perp)$ is atomistic, the atoms being the singleton subsets. The second part follows as well.
\end{proof}

The correspondence between an orthogonality space $(X, \perp)$ and its associated ortholattice ${\mathcal C}(X, \perp)$ extends as follows to automorphisms. Here, an {\it automorphism} of $(X, \perp)$ is a bijection $\phi$ of $X$ such that, for any $x, y \in X$, $\;x \perp y$ if and only if $\phi(x) \perp \phi(y)$. We denote the automorphism group of $(X, \perp)$ by $\Aut{X,\perp}$. Moreover, the group of automorphisms of the ortholattice ${\mathcal C}(X, \perp)$ is denoted by $\Aut{{\mathcal C}(X, \perp)}$.

\begin{proposition} \label{prop:automorphisms-OS-ortholattice}
Let $\phi$ be an automorphism of the orthogonality space $(X, \perp)$. Then
\begin{equation} \label{fml:automorphisms-OS-ortholattice}
\bar\phi \colon {\mathcal C}(X, \perp) \to {\mathcal C}(X, \perp) \komma A \mapsto \{ \phi(e) \colon e \in A\}
\end{equation}
is an automorphism of the ortholattice ${\mathcal C}(X, \perp)$.

If $(X,\perp)$ is point-closed, then $\Aut{X,\perp} \to \Aut{{\mathcal C}(X, \perp)} \komma \phi \mapsto \bar\phi$ is an isomorphism.
\end{proposition}

\begin{proof}
The first part is clear. If the singleton subsets are orthoclosed, then ${{\mathcal C}(X,\perp)}$ is atomistic and consequently, every automorphism is induced by a unique orthogonality-preserving permutation of the atoms. The second part follows as well.
\end{proof}

We now turn to the correspondence between Hermitian spaces and ortholattices; see, e.g., \cite[Section 34]{MaMa}.

For a subset $E$ of a Hermitian space $H$, we define $E\c = \{ u \in H \colon u \perp E \}$. Let $H$ be finite-dimensional. Then $E \subseteq H$ is a subspace of $H$ if and only if $E = E\cc$. We partially order the set ${\mathcal L}(H)$ of subspaces of $H$ w.r.t.\ the set-theoretic inclusion and we endow ${\mathcal L}(H)$ with the complementation function $\c$. Then ${\mathcal L}(H)$ is a complete ortholattice.

We recall that a lattice with $0$ is called {\it atomistic} if each element is the join of atoms. Moreover, we call an ortholattice {\it irreducible} if it is not isomorphic to the direct product of two non-trivial ortholattices. Here, an ortholattice is considered trivial if consisting of a single element.

\begin{theorem} \label{thm:orthomodular-space}
Let $H$ be a Hermitian space of finite dimension $m$. Then ${\mathcal L}(H)$ is an irreducible, atomistic, modular ortholattice of length $m$.

Conversely, let $L$ be an irreducible, atomistic, modular ortholattice of finite length $m \geq 4$. Then there is a $\star$-sfield $K$ and an $m$-dimensional Hermitian space $H$ over $K$ such that $L$ is isomorphic to ${\mathcal L}(H)$.
\end{theorem}

A linear operator $U \colon H \to H$ of a Hermitian space $H$ is called {\it unitary} if $U$ is a linear isomorphism such that $\herm{U(x)}{U(y)} = \herm x y$ for any $x, y \in H$. We denote the group of unitary operators by $\U(H)$ and its identity by $I$. Furthermore, we denote the group of automorphisms of the ortholattice ${\mathcal L}(H)$ by $\Aut{{\mathcal L}(H)}$.

The relationship between the automorphisms of a Hermitian space $H$ and its subspace ortholattice ${\mathcal L}(H)$ is described by Piron's version of Wigner's Theorem \cite[Thm.~3.28]{Pir}; see also \cite{May}. We shall be interested only in those automorphisms of ${\mathcal L}(H)$ that are induced by linear operators.

For a subspace $F$ or $H$, we denote by $[0,F]$ the interval of ${\mathcal L}(H)$ consisting of all subspaces contained in $F$.

\begin{theorem} \label{thm:Wigner}
Let $H$ be a Hermitian space of finite dimension $\geq 3$. For any unitary operator $U$ on $H$, the map
\[ \lambda_U \colon {\mathcal L}(H) \to {\mathcal L}(H) \komma E \mapsto \{ U(x) \colon x \in E \} \]
is an automorphism of ${\mathcal L}(H)$. The map $\U(H) \to \Aut{{\mathcal L}(H)} \komma U \mapsto \lambda_U$ is a homomorphism, whose kernel is $\{ \epsilon I \colon \epsilon \in Z(K) \cap U(K) \}$.

Conversely, let $\lambda$ be an automorphism of ${\mathcal L}(H)$ and assume that there is an at least two-dimensional subspace $F$ such that $\lambda|_{[0,F]}$ is the identity. Then there is a unique unitary operator $U$ on $H$ such that $\lambda = \lambda_U$ and $U|_F$ is the identity.
\end{theorem}

Given a Hermitian space $H$, we deal in this work with automorphisms of $(P(H), \perp)$ rather than ${\mathcal L}(H)$. We will modify Theorem \ref{thm:Wigner} accordingly. Note that the subspaces of $H$ and the orthoclosed subsets of $(P(H), \perp)$ are in a natural one-to-one correspondence; we may in fact identify the ortholattices ${\mathcal L}(H)$ and ${\mathcal C}(P(H), \perp)$. We may thus use Proposition \ref{prop:automorphisms-OS-ortholattice} to get the following further version of Wigner's Theorem, which in the case of a complex Hilbert space is actually Uhlhorn's Theorem \cite{Uhl}.

\begin{theorem} \label{thm:Wigner-for-OS}
Let $H$ be a Hermitian space of finite dimension $\geq 3$. For any unitary operator $U$, the map
\begin{equation} \label{fml:map-induced-by-unitary-operator}
\phi_U \colon P(H) \to P(H) \komma \lin x \mapsto \lin{U(x)}
\end{equation}
is an automorphism of $(P(H), \perp)$. The map $\U(H) \to \Aut{P(H),\perp} \komma U \mapsto \phi_U$ is a homomorphism, whose kernel is $\{ \epsilon I \colon \epsilon \in Z(K) \cap U(K) \}$.

Conversely, let $\phi$ be an automorphism of $(P(H), \perp)$ and assume that there is an at least two-dimensional subspace $F$ of $H$ such that $\phi(\lin x) = \lin x$ for any $x \in F\withoutzero$. Then there is a unique unitary operator $U$ on $H$ such that $\phi = \phi_U$ and $U|_F$ is the identity.
\end{theorem}

\begin{proof}
We may extend any automorphism $\phi$ of $(P(H), \perp)$ to all of ${\mathcal L}(H)$, defining $\bar\phi \colon {\mathcal L}(H) \to {\mathcal L}(H) \komma E \mapsto \bigvee \{ \phi(\lin x) \colon x \in E \}$. By Proposition \ref{prop:automorphisms-OS-ortholattice} and in view of the identification of ${\mathcal C}(P(H), \perp)$ with ${\mathcal L}(H)$, this assignment defines a one-to-one correspondence between the automorphisms of $(P(H),\perp)$ and ${\mathcal L}(H)$. Hence the assertions follow from Theorem \ref{thm:Wigner}.
\end{proof}

For a unitary operator $U$ of a Hermitian space $H$, $\phi_U$ will denote in the sequel the automorphism of $(P(H), \perp)$ induced by $U$ according to (\ref{fml:map-induced-by-unitary-operator}).

\section{The representation by Hermitian spaces}
\label{sec:Hermitian-spaces}

Our first aim is to identify finite-dimensional Hermitian spaces with special orthogonality spaces. In contrast to the procedure in \cite{Vet1}, we do not deal already at this stage with symmetries. We rather derive the structure of a Hermitian space on the basis of two first-order conditions.

Throughout the remainder of this paper, $(X, \perp)$ will always be an irredundant orthogonality space of finite rank. We will call $(X, \perp)$ {\it linear} if the following two conditions are fulfilled:
\begin{itemize}

\item[\rm (L$_1$)] Let $e \in X$. Then for any $f \neq e$ there is a $g \perp e$ such that $\{e,g\}\c = \{e,f\}\c$.

\item[\rm (L$_2$)] Let $e \in X$. Then for any $g \perp e$ there is a $f \neq e, g$ such that $\{e,g\}\c = \{e,f\}\c$.

\end{itemize}
Condition (L$_1$) says that the collection of elements orthogonal to distinct elements $e$ and $f$ can be specified in such a way that $f$ is replaced with an element orthogonal to~$e$. (L$_1$) can be seen as a version of orthomodularity; indeed, this property is among its consequences. But more is true; also atomisticity follows and thus (L$_1$) can be regarded as the key property for the representability of $X$ as a linear space.

Condition (L$_2$) can be regarded as a statement complementary to (L$_1$). Indeed, (L$_2$) says that the collection of elements orthogonal to orthogonal elements $e$ and $g$ can be specified in such a way that $g$ is replaced with a third element. We will actually need only the following immediate consequence of (L$_2$): $\{e,g\}\cc$, where $e \perp g$, is never a two-element set. As we will see below, a closely related property of $(X, \perp)$ is its irreducibility.

\begin{example} \label{ex:orthogonality-space-from-orthomodular-space-is-linear}
Let $H$ be a finite-dimensional Hermitian space. Then $(P(H), \perp)$ is linear. To see that $P(H)$ fulfils {\rm (L$_1$)}, let $x, y \in H\withoutzero$. Putting $z = y - \herm{y}{x}\herm{x}{x}^{-1} x$, we have $z \perp x$ and $\lin{x,y} = \lin{x,z}$. In particular then, $\{x,y\}\c = \{x,z\}\c$. Also condition {\rm (L$_2$)} is immediate. Indeed, if $x, y \in H\withoutzero$ such that $x \perp y$, we have that $\lin{x,y} = \lin{x,x+y}$. In particular then, $\{x,y\}\c = \{x,x+y\}\c$.
\end{example}

\begin{lemma} \label{lem:atomistic}
Let $(X, \perp)$ fulfil {\rm (L$_1$)}. Then $(X, \perp)$ is point-closed. In particular, ${{\mathcal C}(X, \perp)}$ is atomistic, the atoms being the singletons $\{e\}$, $e \in X$.

Moreover, the assignment $X \to {\mathcal C}(X, \perp) \komma e \mapsto \{e\}$ defines an isomorphism between $(X, \perp)$ and the set of atoms of ${\mathcal C}(X, \perp)$ endowed with the inherited orthogonality relation.
\end{lemma}

\begin{proof}
Let $e \in X$ and $f \in \{e\}\cc$. Then $\{f\}\cc \subseteq \{e\}\cc$ and hence $\{e\}\c \subseteq \{f\}\c$. Assume $e \neq f$. Then there is, by (L$_1$), a $g \perp e$ such that $\{e,g\}\c = \{e,f\}\c$. It follows that $g \in \{e\}\c = \{e,f\}\c = \{e,g\}\c$, a contradition. Hence $f = e$, and we conclude that $\{e\}\cc = \{e\}$.

The first part follows, the second part holds by Proposition \ref{prop:OS-ortholattice}.
\end{proof}

We call a subset $D$ of $X$ {\it orthogonal} if $D$ consists of pairwise orthogonal elements.

\begin{lemma} \label{lem:for-Dacey}
Let $(X, \perp)$ fulfil {\rm (L$_1$)}. Let $D \subseteq X$ be orthogonal and let $e \notin D\cc$. Then there is an $f \perp D$ such that $(D \cup \{e\})\cc = (D \cup \{f\})\cc$.
\end{lemma}

\begin{proof}
The assertion is trivial if $D$ is empty; let us assume that $D$ is non-empty. As we have assumed $X$ to have finite rank, $D$ is finite. Let $D = \{d_1, \ldots, d_k \}$, where $k \geq 1$.

By (L$_1$), there is an $e_1 \perp d_1$ such that $\{d_1,e\}\cc = \{d_1,e_1\}\cc$. Similarly, we see that there is, for $i = 2, \ldots, k$, an $e_i \perp d_i$ such that $\{d_i,e_{i-1}\}\cc = \{d_i,e_i\}\cc$. We conclude
\[ \begin{split}
(D \cup \{e\})\cc
& \;=\; \{e\} \vee \{d_1\} \vee \ldots \vee \{d_k\} \\
& \;=\; \{d_1\} \vee \{e_1\} \vee \{d_2\} \vee \ldots \vee \{d_k\} \\
& \;=\; \{d_1\} \vee \{d_2\} \vee \{e_2\} \vee \ldots \vee \{d_k\} \\
& \;=\; \ldots \\
& \;=\; \{d_1\} \vee \{d_2\} \vee \ldots \vee \{d_k\} \vee \{e_k\}
\;=\; (D \cup \{e_k\})\cc.
\end{split} \]
We observe that $f = e_k$ fulfils the requirement.
\end{proof}

The following useful criterion for ${\mathcal C}(X, \perp)$ to be orthomodular is due to J.~R.~Dacey \cite{Dac}; see also \cite[Theorem~35]{Wlc}.

\begin{lemma} \label{lem:Dacey}
${\mathcal C}(X, \perp)$ is orthomodular if and only if, for any $A \in {\mathcal C}(X, \perp)$ and any maximal orthogonal subset $D$ of $A$, we have $A = D\cc$.
\end{lemma}

It follows that, by virtue of condition (L$_1$), we may describe ${\mathcal C}(X, \perp)$ as follows.

\begin{lemma} \label{lem:MOL}
Let $(X, \perp)$ fulfil {\rm (L$_1$)} and let $m$ be the rank of $X$. Then ${\mathcal C}(X, \perp)$ is an atomistic, modular ortholattice of length $m$.
\end{lemma}

\begin{proof}
By Proposition~\ref{prop:OS-ortholattice} and Lemma~\ref{lem:atomistic}, ${\mathcal C}(X, \perp)$ is an atomistic ortholattice. From Lemmas~\ref{lem:for-Dacey} and~\ref{lem:Dacey}, it follows that ${\mathcal C}(X, \perp)$ is orthomodular.

As we have assumed $X$ to be of finite rank $m$, the top element $X$ of ${\mathcal C}(X, \perp)$ is the join of $m$ mutually orthogonal atoms. It follows that ${\mathcal C}(X, \perp)$ has length $m$.

We claim that ${\mathcal C}(X, \perp)$ fulfils the covering property. Let $A \in {\mathcal C}(X, \perp)$ and let $e \in X$ be such that $e \notin A$. By Lemma~\ref{lem:Dacey}, there is an orthogonal set $D$ such that $A = D\cc$. By Lemma~\ref{lem:for-Dacey}, there is an $f \perp D$ such that $A \vee \{e\} = (D \cup \{e\})\cc = (D \cup \{f\})\cc = A \vee \{f\}$. Note that $\{f\}$ is an atom orthogonal to $A$. Hence it follows by the orthomodularity of ${\mathcal C}(X, \perp)$ that $A \vee \{e\}$ covers $A$.

Finally, an atomistic ortholattice of finite length fulfilling the covering property is modular \cite[Lemma~30.3]{MaMa}.
\end{proof}

We now turn to the consequences of condition (L$_2$). In the presence of (L$_1$), there are a couple of alternative formulations.

We call $(X, \perp)$ {\it reducible} if $X$ is the disjoint union of non-empty sets $A$ and $B$ such that $e \perp f$ for any $e \in A$ and $f \in B$, and otherwise {\it irreducible}.

\begin{lemma} \label{lem:irreducible}
Let $(X, \perp)$ fulfil {\rm (L$_1$)}. Then the following are equivalent:
\begin{itemize}

\item[\rm (1)] $X$ fulfils {\rm (L$_2$)}.

\item[\rm (2)] For any orthogonal elements $e, f \in X$, $\{e,f\}\cc$ contains a third element.

\item[\rm (3)] $X$ is irreducible.

\item[\rm (4)] ${\mathcal C}(X)$ is irreducible.

\end{itemize}
\end{lemma}

\begin{proof}
(1) $\Rightarrow$ (2): This is obvious.

(2) $\Rightarrow$ (1): Assume that (2) holds. Let $e$ and $g$ be orthogonal elements of $X$. By assumption, $\{e,g\}\cc$ contains a third element $f$. Then $\{f\} \subseteq \{e\} \vee \{g\}$ and $\{e\} \cap \{f\} = \emptyset$. By Lemma \ref{lem:MOL}, ${\mathcal C}(X)$ fulfils the exchange property, hence $\{e,f\}\cc = \{e\} \vee \{f\} = \{e\} \vee \{g\} = \{e,g\}\cc$ and we conclude $\{e,f\}\c = \{e,g\}\c$. We have shown (L$_2$).

(2) $\Rightarrow$ (3): Assume that $X$ is reducible. Then $X = A \cup B$, where $A$ and $B$ are disjoint non-empty sets such that $e \perp f$ for any $e \in A$ and $f \in B$. Pick $e \in A$ and $f \in B$ and let $g \in \{e,f\}\cc$. We have that either $g \in A$ or $g \in B$. In the former case, $g \perp f$ and hence $g \in \{e,f\}\cc \cap \{f\}\c = (\{e\} \vee \{f\}) \cap \{f\}\c = \{e\}$, that is, $g = e$. Similarly, in the latter case, we have $g = f$. We conclude that $\{e,f\}\cc$ contains two elements only.

(3) $\Rightarrow$ (4): Assume that ${\mathcal C}(X)$ is not irreducible. Then ${\mathcal C}(X)$ is the direct product of non-trivial ortholattices $L_1$ and $L_2$. The atoms of $L_1 \times L_2$ are of the form $(p,0)$ or $(0,q)$, for an atom $p$ of $L_1$ or an atom $q$ of $L_2$, respectively. Furthermore, $(a,0) \perp (0,b)$ for any $a \in L_1$ and $b \in L_2$. We conclude that the set of atoms of ${\mathcal C}(X, \perp)$ can be partitioned into two non-empty subsets such that any element of one set is orthogonal to any of the other one. In view of Lemma \ref{lem:atomistic}, we conclude that $(X,\perp)$ is reducible.

(4) $\Rightarrow$ (2): Assume that ${\mathcal C}(X)$ is irreducible. By \cite[Theorem 13.6]{MaMa}, below the join of any two atoms of ${\mathcal C}(X, \perp)$ there is a third atom. In particular, for orthogonal elements $e, f \in X$, $\{e,f\}\cc = \{e\} \vee \{f\}$ contains a third element.
\end{proof}

We summarise:

\begin{theorem} \label{thm:OS}
Let $(X, \perp)$ be linear and of finite rank $m$. Then ${\mathcal C}(X, \perp)$ is an irreducible, atomistic, modular ortholattice of length $m$.
\end{theorem}

We arrive at the main result of this section.

\begin{theorem} \label{thm:orthogonality-spaces-by-orthomodular-spaces}
Let $(X, \perp)$ be a linear orthogonality space of finite rank $m \geq 4$. Then there is a $\star$-sfield $K$ and an $m$-dimensional Hermitian space $H$ over $K$ such that ${\mathcal C}(X, \perp)$ is isomorphic to ${\mathcal L}(H)$. In particular, $(X,\perp)$ is then isomorphic to ${(P(H), \perp)}$.
\end{theorem}

\begin{proof}
By Theorems \ref{thm:OS} and \ref{thm:orthomodular-space}, there is an $m$-dimensional Hermitian space $H$ such that ${\mathcal C}(X, \perp)$ is isomorphic to ${\mathcal L}(H)$. Moreover, by Lemma \ref{lem:atomistic}, $(X, \perp)$ can be identified with the set of atoms of ${\mathcal C}(X, \perp)$ endowed with the inherited orthogonality relation; and the same applies to $(P(H), \perp)$ and ${\mathcal L}(H)$.
\end{proof}

\section{The representation by quadratic spaces}
\label{sec:ordered-fields}

Provided that the rank is finite and at least $4$, we have seen that a linear orthogonality space arises from a Hermitian space over some $\star$-sfield. Our objective is to investigate the consequences of an additional condition. It will turn out that we can specify the $\star$-sfield considerably more precisely, namely, as a (commutative) formally real field.

We shall now make precise our idea to which we refer as the gradual transitivity of the orthogonality space. Given distinct elements $e$ and $f$, we will require a divisible group of automorphisms to exist such that the group orbit of $e$ is exactly $\{e,f\}\cc$ and $\{e,f\}\c$ is kept pointwise fixed.

It seems natural to assume that the group is, at least locally, linearly parametrisable. By the following lemma, the automorphism that maps $e$ to some $f \perp e$ actually interchanges $e$ and $f$. Accordingly, we will postulate that the group is cyclically ordered.

\begin{lemma} \label{lem:transition-between-orthogonal-elements-is-exchange}
Let $(X, \perp)$ be linear and of rank $\geq 4$. Let $e, f \in X$ such that $e \perp f$. Let $\phi$ be an automorphism of $X$ such that $\phi(e) = f$ and $\phi(d) = d$ for any $d \perp e, f$. Then $\phi(f) = e$.
\end{lemma}

\begin{proof}
In accordance with Theorem \ref{thm:orthogonality-spaces-by-orthomodular-spaces}, let $H$ be the Hermitian space such that we can identify $(X, \perp)$ with $(P(H), \perp)$. Let $u, v \in H\withoutzero$ be such that $e = \lin u$ and $f = \lin v$. By Theorem \ref{thm:Wigner-for-OS}, there is a unitary operator $U$ inducing $\phi$ and being the identity on $\{u,v\}\c$. Then $U(u) \in \lin v$ and $U(w) = w$ for any $w \perp \lin{u,v}$. Hence $U(v) \in \lin u$, that is, $\phi(f) = e$.
\end{proof}

In what follows, we write $\Reals/2\pi\Integers$ for the additive group of reals modulo $\{2k\pi \colon k \in \Integers\}$, which can be identified with the circle group, that is, with the multiplicative group of complex numbers of modulus $1$. Moreover, let $G$ be a group of bijections of some set $W$, and let $S \subseteq W$. Then we say that $G$ {\it acts on $S$ transitively} if $S$ is invariant under $G$ and the action of $G$ restricted to $S$ is transitive. Moreover, we say that $G$ {\it acts on $S$ trivially} if, for all $g \in G$, $g$ is the identity on $S$.

We define the following condition on $(X, \perp)$. Here, we call an orthoclosed subset of the form $\{e,f\}\cc$, where $e$ and $f$ are distinct elements of $X$, a {\it line}.
\begin{itemize}

\item[(R$_1$)] For any line $L \subseteq X$, there is a divisible subgroup $\CircleGroup$ of $\Reals/2\pi\Integers$ and an injective homomorphism $\kappa \colon \CircleGroup \to \Aut{X,\perp} \komma t \mapsto \kappa_t$ such that,
\begin{itemize}

\item[($\alpha$)] the group $\{ \kappa_t \colon t \in \CircleGroup \}$ acts on $L$ transitively;

\item[($\beta$)] the group $\{ \kappa_t \colon t \in \CircleGroup \}$ acts on $L\c$ trivially.

\end{itemize}

\end{itemize}

Our discussion will focus to a large extent on the symmetries of $(X, \perp)$ that are described in condition (R$_1$). We will use the following terminology. For a line $L$, let $\kappa \colon \CircleGroup \to \Aut{X,\perp}$ be as specified in condition (R$_1$). Then we call an automorphism $\kappa_t$, $t \in \CircleGroup$, a {\it basic circulation} in $L$ and we call the subgroup $\{ \kappa_t \colon t \in \CircleGroup \}$ of $\Aut{X,\perp}$ a {\it basic circulation group} of $L$. Note that, by the injectivity requirement in condition (R$_1$), this group is isomorphic to $\CircleGroup$.

Moreover, we denote by $\Circ{X,\perp}$ the subgroup of $\Aut{X,\perp}$ that is generated by all basic circulations. The automorphisms belonging to $\Circ{X,\perp}$ are called {\it circulations} and $\Circ{X,\perp}$ itself is the {\it circulation group}.

\begin{example}
Let the $\Reals^n$, for a finite $n \geq 1$, be endowed with the usual Euclidean inner product. Then $(P(\Reals^n), \perp)$ is a linear orthogonality space fulfilling {\rm (R$_1$)}. Indeed, let $u, v$ be an orthonormal basis of a $2$-dimensional subspace of $\Reals^n$. Let $\CircleGroup = \Reals/2\pi\Integers$ and let $\kappa_t$, $t \in \CircleGroup$, be the rotation in the (oriented) $u$-$v$-plane by the angle $e^{it}$ and the identity on $\lin{u,v}\c$. Then conditions {\rm ($\alpha$)} and {\rm ($\beta$)} are obviously fulfilled.
\end{example}

For the general case, the intended effect of condition (R$_1$) is described in the following lemma. For $\phi \in \Aut{X,\perp}$ and $n \geq 1$, we let $\phi^n = \phi \circ \ldots \circ \phi$ ($n$ factors).

\begin{lemma} \label{lem:transitivity}
Let $(X, \perp)$ be of rank $\geq 4$, linear, and fulfilling {\rm (R$_1$)}. Let $e$ and $f$ be distinct elements of $X$. Then for each $n \geq 1$ there is an automorphism $\phi$ of $(X, \perp)$ such that $\phi^n(e) = f$ and $\phi(d) = d$ for any $d \perp e,f$. In case when $e$ and $f$ are orthogonal, we have in addition that $\phi^n(f) = e$.
\end{lemma}

\begin{proof}
By (R$_1$), applied to $\{e,f\}\cc$, there is a divisible subgroup $C$ of $\Aut{X,\perp}$ that acts transitively on $\{e,f\}\cc$ and is the identity on $\{e,f\}\c$. In particular, there is a $\psi \in C$ such that $\psi(e) = f$ and, by the divisibility of $C$, there is for any $n \geq 1$ a $\phi \in C$ such that $\phi^n = \psi$. The first part is clear; the additional assertion follows from Lemma~\ref{lem:transition-between-orthogonal-elements-is-exchange}.
\end{proof}

Our aim is to investigate the consequences of condition (R$_1$) for a linear orthogonality space. We first mention that (L$_2$), as part of the conditions of linearity, becomes redundant.

\begin{lemma} \label{lem:irreducibility-from-CG}
Let $(X, \perp)$ fulfil {\rm (L$_1$)} and {\rm (R$_1$)}. Then $X$ fulfils {\rm (L$_2$)}, that is, $X$ is linear.
\end{lemma}

\begin{proof}
Let $e, f \in X$ be orthogonal. We will show that $\{e,f\}\cc$ contains a third element. The assertion will then follow from Lemma \ref{lem:irreducible}.

Assume to the contrary that $\{e,f\}\cc$ is a two-element set. Let $\{ \kappa_t \colon t \in \CircleGroup \}$ be a basic circulation group of $\{e,f\}$. As the group acts transitively on $\{e,f\}$, there is a $t \in \CircleGroup \setminus \{0\}$ such that $\kappa_t(e) = f$. But $\{e,f\}$ is invariant also under $\kappa_{\frac t 2}$ and we have $\kappa_{\frac t 2}^2 = \kappa_t$, an impossible situation.
\end{proof}

The transitivity of a linear orthogonality space, which by Lemma \ref{lem:transitivity} is a consequence of condition (R$_1$), allows us to subject the representing Hermitian space to an additional useful condition.

\begin{lemma} \label{lem:unit-vectors}
Let $(X, \perp)$ be linear, of rank $\geq 4$, and fulfilling {\rm (R$_1$)}. Then there is a Hermitian space $H$ such that $(X, \perp)$ is isomorphic to $(P(H), \perp)$ and such that each one-dimensional subspace contains a unit vector.
\end{lemma}

\begin{proof}
By Theorem \ref{thm:orthogonality-spaces-by-orthomodular-spaces}, there is a Hermitian space $H$ such that $(X, \perp)$ is isomorphic to $(P(H), \perp)$.

Let $u \in H$. We can define a new Hermitian form on $H$ inducing the same orthogonality relation and such that $u$ becomes a unit vector; see, e.g., \cite{Hol3}. By Lemma \ref{lem:transitivity} and Theorem \ref{thm:Wigner-for-OS}, there is for any $v \in H$ a unitary operator such that $U(u) \in \lin v$. The assertion follows.
\end{proof}

For the rest of this section, let $H$ be a Hermitian space over the $\star$-sfield $K$ such that $H$ is of finite dimension $\geq 4$, each one-dimensional subspace contains a unit vector, and $(P(H), \perp)$ fulfils (R$_1$). Our aim is to be as specific as possible about the $\star$-sfield $K$.

\begin{lemma} \label{lem:circulation-group}
Let $T$ be a $2$-dimensional subspace of $H$ and let $\{ \kappa_t \colon t \in \CircleGroup \}$ be a basic circulation group of $P(T)$. Then, for each $t \in \CircleGroup$, there is a uniquely determined unitary operator $U_t$ inducing $\kappa_t$ and being the identity on $T\c$. Moreover, $ \CircleGroup \to U(H) \komma t \mapsto U_t$ is an injective homomorphism.
\end{lemma}

\begin{proof}
By Theorem \ref{thm:Wigner-for-OS}, $\kappa_t$ is, for each $t \in \Reals$, induced by a unique unitary operator $U_t$ such that $U_t|_{T\c}$ is the identity. In particular, $\kappa_0$ is the identity on $P(H)$, hence $U_0$ must be the identity on $H$. Furthermore, for any $s, t \in \CircleGroup$, $U_s U_t$ induces $\kappa_{s+t} = \kappa_s \, \kappa_t$ and is the identity on $T\c$. The same applies to $U_{s+t}$ and it follows that $U_{s+t} = U_s U_t$. Finally, the injectivity assertion follows from the fact that, according to (R$_1$), the assignment $t \mapsto \kappa_t$ is already injective.
\end{proof}

\begin{lemma} \label{lem:K-is-field}
$K$ is commutative and the involution $^\star$ is the identity. In particular, $H$ is a quadratic space.
\end{lemma}

\begin{proof}
Let $T$ be a two-dimensional subspace of $H$. Let $\{ \kappa_t \colon t \in \CircleGroup \}$ be a basic circulation group of $P(T)$ and, in accordance with Lemma \ref{lem:circulation-group}, let the unitary operator $U_t$, for each $t \in \CircleGroup$, induce $\kappa_t$.

We will identify the operators $U_t$, $t \in \CircleGroup$, with their restriction to $T$ and represent them, w.r.t.\ a fixed  orthonormal basis $b_1, b_2$ of $T$, by $2 \times 2$-matrices. Let $t \in \CircleGroup$. Then $U_t = \begin{smm} \alpha & \gamma \\ \beta & \delta \end{smm}$, where $\alpha \alpha^\star + \beta \beta^\star = \gamma \gamma^\star + \delta \delta^\star = 1$ and $\alpha \gamma^\star + \beta \delta^\star = 0$. As $\kappa(\CircleGroup)$ acts transitively on $P(T)$, there is a $p \in \CircleGroup$ such that $U_p(\lin{b_1}) = \lin{b_2}$ and consequently also $U_p(\lin{b_2}) = \lin{b_1}$. Hence $U_p = \begin{smm} 0 & \epsilon_1 \\ \epsilon_2 & 0 \end{smm}$ for some $\epsilon_1, \epsilon_2 \in U(K)$.

Because
\[ \begin{split}
\begin{smm} \epsilon_2 \gamma & \epsilon_1 \alpha \\ \epsilon_2 \delta & \epsilon_1 \beta \end{smm}
\;=\; &\begin{smm} \alpha & \gamma \\ \beta & \delta \end{smm} \cdot
\begin{smm} 0 & \epsilon_1 \\ \epsilon_2 & 0 \end{smm}
\;=\; U_t \, U_p \;=\; U_{t+p} \\
\;=\; U_p \, U_t \;=\; &\begin{smm} 0 & \epsilon_1 \\ \epsilon_2 & 0 \end{smm} \cdot
\begin{smm} \alpha & \gamma \\ \beta & \delta \end{smm}
\;=\; \begin{smm} \beta \epsilon_1 & \delta \epsilon_1 \\ \alpha \epsilon_2 & \gamma \epsilon_2 \end{smm},
\end{split} \]
we have
\begin{equation} \label{fml:Ut1}
U_t \;=\;
\begin{smm} \alpha & \epsilon_1 \beta \epsilon_2^\star \\ \beta & \epsilon_1 \alpha \epsilon_1^\star \end{smm}
\;=\; \begin{smm} \alpha & \epsilon_2^\star \beta \epsilon_1 \\ \beta & \epsilon_2^\star \alpha \epsilon_2 \end{smm}.
\end{equation}

We next claim that, for any $\xi \in K$, there is a $t \in \CircleGroup$ such that $\xi = \beta^{-1} \alpha$, where $\begin{smm} \alpha \\ \beta \end{smm}$ is the first column vector of $U_t$. Indeed, by the transitivity of $\kappa(\CircleGroup)$, there is a $t \in \CircleGroup$ such that $U_t = \begin{smm} \alpha & \gamma \\ \beta & \delta \end{smm}$ maps $\lin{e_1}$ to $\lin{\xi e_1 + e_2}$. Then $\beta \neq 0$ and $\lin{\begin{smm} \beta^{-1}\alpha \\ 1 \end{smm}} = \lin{\begin{smm} \alpha \\ \beta \end{smm}} = \lin{U(e_1)} = \lin{\begin{smm} \xi \\ 1 \end{smm}}$, thus the assertion follows.

The orthogonality of the column vectors of the first matrix in (\ref{fml:Ut1}) implies $\alpha \epsilon_2 \beta^\star \epsilon_1^\star + \beta \epsilon_1 \alpha^\star \epsilon_1^\star = 0$ and hence $(\beta^{-1} \alpha)^\star = -\epsilon_1^\star \beta^{-1} \alpha \epsilon_2$, provided that $\beta \neq 0$. By the previous remark, we conclude $\xi^\star = -\epsilon_1^\star \xi \epsilon_2$ for any $\xi \in K$. From the case $\xi = 1$ we see that $\epsilon_2 = -\epsilon_1$. Let $\epsilon = \epsilon_2$. Then $\epsilon \in U(K)$ is such that
\begin{equation} \label{fml:xi-star}
\xi^\star \;=\; \epsilon^\star \xi \epsilon \text{ for any $\xi \in K$},
\end{equation}
and we conclude that for each $t \in \CircleGroup$ there are $\alpha, \beta \in K$ such that
\begin{equation} \label{fml:Ut2}
U_t \;=\; \begin{smm} \alpha & -\beta^\star \\ \beta & \alpha^\star \end{smm}.
\end{equation}
Let now $s \in \CircleGroup$ be such that $U_s$ maps $\lin{e_1}$ to $\lin{e_1+e_2}$. Then there is a $\gamma \in K$ such that $U_s = \begin{smm} \gamma & -\gamma^\star \\ \gamma & \gamma^\star \end{smm}$. Note that $2 \gamma \gamma^\star = 1$; in particular, $K$ does not have characteristic~$2$. Moreover, given any $U_t$ according to (\ref{fml:Ut2}), we have
\[ \begin{split}
\begin{smm} \tilde\alpha & -\tilde\beta^\star \\ \tilde\beta & \tilde\alpha^\star \end{smm} \;=\;
& \begin{smm} \gamma & -\gamma^\star \\ \gamma & \gamma^\star \end{smm} \cdot
\begin{smm} \alpha & -\beta^\star \\ \beta & \alpha^\star \end{smm}
\;=\; U_s \, U_t \\
\;=\; & U_t \, U_s
\;=\;
\begin{smm} \alpha & -\beta^\star \\ \beta & \alpha^\star \end{smm} \cdot
\begin{smm} \gamma & -\gamma^\star \\ \gamma & \gamma^\star \end{smm}.
\end{split} \]
This means
\begin{align*}
& \alpha \gamma - \beta \gamma^\star \;=\;
  \gamma \alpha - \gamma \beta^\star \;=\; \tilde\alpha, \\
& \alpha \gamma + \beta \gamma^\star \;=\;
  \gamma \beta + \gamma \alpha^\star \;=\; \tilde\beta, \\
& \gamma \alpha + \gamma^\star \beta \;=\; 
  \alpha^\star \gamma + \beta \gamma \;=\; \tilde\beta, \\
& \gamma \alpha - \gamma^\star \beta \;=\;
  \alpha \gamma - \beta^\star \gamma \;=\; \tilde\alpha.
\end{align*}
Consequently, $2 \alpha \gamma = 2 \gamma \alpha = \tilde\alpha + \tilde\beta$ and $2 \beta \gamma^\star = 2 \gamma^\star \beta = \tilde\beta - \tilde\alpha$. Hence $\gamma$ commutes with $\alpha$ and, because $2\gamma^\star = \gamma^{-1}$, also with $\beta$. We conclude that $\gamma \in Z(K)$. By (\ref{fml:xi-star}), it follows that $\gamma^\star = \gamma$. Furthermore, we have $(\alpha + \beta + \alpha^\star - \beta^\star) \gamma = \tilde\alpha + \tilde\beta = 2 \alpha \gamma$ and $(\beta - \alpha + \alpha^\star + \beta^\star) \gamma = \tilde\beta - \tilde\alpha = 2 \beta \gamma$. It follows that $\alpha^\star - \beta^\star = \alpha - \beta$ and $\alpha^\star + \beta^\star = \alpha + \beta$, that is, $\alpha = \alpha^\star$ and $\beta = \beta^\star$.

Since $\alpha = \alpha^\star = \epsilon^\star \alpha \epsilon$, we have $\alpha \epsilon = \epsilon \alpha$, and similarly we see that $\beta \epsilon = \epsilon \beta$. Hence $(\beta^{-1} \alpha)^\star = \epsilon^\star \beta^{-1} \alpha \epsilon = \beta^{-1} \alpha$, provided that $\beta \neq 0$. We conclude $\xi^\star = \xi$ for any $\xi \in K$. That is, the involution is the identity, and the $\star$-sfield is commutative.
\end{proof}

We continue by showing that $K$ can be endowed with an ordering to the effect that the quadratic space $H$ becomes positive definite. We refer to \cite[\S 1]{Pre} for further information on the topic of fields and orderings.

\begin{lemma} \label{lem:K-is-ordered}
$K$ is a formally real field. W.r.t.~any order on $K$, the hermitian form on $H$ is positive definite.
\end{lemma}

\begin{proof}
Let
\[ S_K \;=\;
       \{ \alpha_1^2 + \ldots + \alpha_k^2 \colon
       \alpha_1, \ldots, \alpha_k \in K, \; k \geq 0 \} \]
and note that, if $K$ admits an order, then all elements of $S_K$ will be positive. We shall show that $S_K \cap -S_K = \{0\}$; it then follows that $S_K$ can be extended to a positive cone determining an order that makes $K$ into an ordered field; see, e.g., \cite[Theorem~(1.8)]{Pre}.

Assume to the contrary that $S_K \cap -S_K$ contains a non-zero element. Then there are $\alpha_1, \ldots, \alpha_k \in K \setminus \{0\}$, $k \geq 1$, such that $\alpha_1^2 + \ldots + \alpha_k^2 = 0$.

It follows that that there are non-zero vectors $v_1, \ldots, v_k$ such that $\herm{v_i}{v_i} = \alpha_1^2 + \ldots + \alpha_i^2$, $i = 1, \ldots, k$. Indeed, let $u$ be any unit vector. Then $v_1 = \alpha_1 u$ is non-zero and of length $\alpha_1^2$. Moreover, let $1 \leq i < k$ and assume that $v_i$ is non-zero and of length $\alpha_1^2 + \ldots + \alpha_i^2$. Let $u'$ be a unit vector orthogonal to $v_i$. Then $v_{i+1} = v_i + \alpha_{i+1} u'$ is again non-zero and has length $\alpha_1^2 + \ldots + \alpha_{i+1}^2$.

We conclude that, in particular, there is non-zero vector $v_k$ that has length $\alpha_1^2 + \ldots + \alpha_k^2 = 0$. But this contradicts the anisotropy of the form.

To show also the second assertion, let us fix an order of $K$ and let $v \in H\withoutzero$. Then there is a unit vector $u \in H$ and an $\alpha \in K$ such that $v = \alpha u$. It follows $\herm v v = \herm{\alpha u}{\alpha u} = \alpha^2 > 0$.
\end{proof}

We summarise what we have shown.

\begin{theorem} \label{thm:ordered-field}
Let $(X, \perp)$ be a linear orthogonality space of finite rank $\geq 4$ that fulfils {\rm (R$_1$)}. Then there is an ordered field $K$ and a positive-definite quadratic space $H$ over $K$, possessing unit vectors in each one-dimensional subspace, such that $(X, \perp)$ is isomorphic to $(P(H), \perp)$.
\end{theorem}

We conclude the section with a comment on the formulation of our condition (R$_1$).

\begin{remark}
For the proof of Theorem \ref{thm:ordered-field}, we have not made use of the the divisibility condition in {\rm (R$_1$)}, which hence could be dropped. So far, only Lemma \ref{lem:irreducibility-from-CG}, which we did not use in the sequel, has depended on the divisibility.

We think, however, that it is natural to include this property as it well reflects the idea of gradual transitions between pairs of elements of an orthogonality space. Furthermore, omitting divisibility would be especially interesting if $\CircleGroup$ could possibly be finite. But this is not the case. Indeed, the field of scalars $K$ of the representing linear space has characteristic $0$ and hence each two-dimensional subspace contains infinitely many one-dimensional subspaces. Hence $\CircleGroup$ is necessarily infinite and thus anyhow ``dense'' in $\Reals/2\pi\Integers$.
\end{remark}

\section{The circulation group}
\label{sec:circulation-group}

We have established that linear orthogonality spaces of rank at least $4$ arise from positive definite quadratic spaces in case condition (R$_1$) is fulfilled. We insert a short discussion of the symmetries that are required to exist as part of (R$_1$).

In this section, $H$ will be a positive definite quadratic space over an ordered field $K$ such that $H$ is of finite dimension $\geq 4$, each one-dimensional subspace contains a unit vector, and $(P(H), \perp)$ fulfils (R$_1$). For further information on quadratic spaces, we may refer, e.g., to \cite{Sch}.

In accordance with the common practice, we call the unitary operators of $H$ from now on {\it orthogonal} and we denote the group of orthogonal operators by $\O(H)$. Furthermore, with any endomorphism $A$ of $H$ we may associate its determinant $\det A$. For an orthogonal operator $U$, we have $\det U \in \{1,-1\}$ and we call $U$ a {\it rotation} if $\det U = 1$. The group of rotations is denoted by $\SO(H)$. For a two-dimensional subspace $T$ of $H$, we call $U \in \SO(H)$ a {\it basic rotation} in $T$ if $U|_{T\c}$ is the identity, and we denote the group of basic rotations in $T$ by $\SO(T,H)$.

As should be expected, the basic circulations correspond to the basic rotations.

\begin{proposition} \label{prop:circulation-group-1}
Let $T$ be a two-dimensional subspace of $H$ and let $C$ be a basic circulation group of $P(T)$. Then $C = \{ \phi_U \colon U \in \SO(T,H) \}$ and the map $\SO(T,H) \to C \komma U \mapsto \phi_U$ is an isomorphism.

In particular, there is a unique basic circulation group of $P(T)$. Moreover, any two basic circulation groups are isomorphic.
\end{proposition}

\begin{proof}
In accordance with Lemma \ref{lem:circulation-group}, let $\{ U_t \colon t \in \CircleGroup \}$ be the subgroup of $O(H)$ such that $C = \{ \phi_{U_t} \colon t \in \CircleGroup \}$. We have to show that $\{ U_t \colon t \in \CircleGroup \}$ coincides with $\SO(T,H)$.

As for any $t \in \CircleGroup$ we have $U_t = (U_{\frac t 2})^2$, it is clear that $U_t \in \SO(T,H)$. Conversely, let $U \in \SO(T,H)$. We again fix an orthonormal basis of $T$ and identify the operators in question with the matrix representation of their restriction to $T$. Then we have $U = \begin{smm} \alpha & -\beta \\ \beta & \alpha \end{smm}$ for some $\alpha, \beta \in K$ such that $\alpha^2 + \beta^2 = 1$. As $C$ acts transitively on $P(T)$, there is a $t \in \CircleGroup$ such that $U_t(\begin{smm} 1 \\ 0 \end{smm}) \in \lin{\begin{smm} \alpha \\ \beta \end{smm}}$. This means that $U_t$ equals one of
\[ \begin{smm} \alpha & -\beta \\ \beta & \alpha \end{smm} \text{ or }
\begin{smm} -\alpha & \beta \\ -\beta & -\alpha \end{smm}. \]
Furthermore, we have $U_0 = \begin{smm} 1 & 0 \\ 0 & 1 \end{smm}$ and from ${U_\pi}^2 = U_0$ it follows that $U_\pi = U_0$ or $U_\pi = -U_0$. Since by the injectivity requirement in (R$_1$) the first possibility cannot apply, we have $U_\pi = -U_0 = \begin{smm} -1 & 0 \\ 0 & -1 \end{smm}$. Hence either $U = U_t$ or $U = U_t U_\pi = U_{t+\pi}$. The assertion follows and we conclude that $C = \{ \phi_U \colon U \in \SO(T,H) \}$.

By Lemma \ref{lem:circulation-group}, we thus have the isomorphism $\CircleGroup \to \SO(T,H) \komma t \mapsto U_t$. Moreover, $\CircleGroup \to C \komma t \mapsto \kappa_t$ is an isomorphism, and $\kappa_t = \phi_{U_t}$ for any $t \in \CircleGroup$. We conclude that $\SO(T,H) \to C \komma U \mapsto \phi_U$ is an isomorphism.

The first part as well as the uniqueness assertion is shown. Finally, any two groups $\SO(T,H)$ and $\SO(T',H)$, where $T$ and $T'$ are $2$-dimensional subspaces of $H$, are isomorphic, hence the final assertion follows as well.
\end{proof}

Given a line $L$ in $(P(H),\perp)$, we can speak, in view of Proposition \ref{prop:circulation-group-1}, of {\it the} basic circulation group of $L$. We should note however that, in contrast to the statements on uniqueness and isomorphy in Proposition \ref{prop:circulation-group-1}, the homomorphism from a subgroup $\CircleGroup$ of $\Reals/2\pi\Integers$ to a basic circulation group is not uniquely determined. Indeed, the group $\CircleGroup$ may possess an abundance of automorphisms, as is the case, e.g., for $\CircleGroup = \Reals/2\pi\Integers$.

In Proposition \ref{prop:circulation-group-1}, we have characterised the basic circulation groups as subgroups of $\SO(H)$. We may do so also with respect to the orthogonality space itself.

\begin{lemma} \label{lem:circulation-group-3}
Let $L \subseteq P(H)$ be a line. Then the basic circulation group of $L$ consists of all automorphisms $\phi$ of $(P(H),\perp)$ such that $\phi|_{L\c}$ is the identity and $\phi|_L$ is either the identity or does not possess any fixpoint.
\end{lemma}

\begin{proof}
Let $C$ be the basic circulation group of $L$, and let $T$ be the $2$-dimensional subspace of $H$ such that $L = P(T)$.

Let $\phi \in C$. By Proposition \ref{prop:circulation-group-1}, $\phi$ is induced by some $U \in \SO(T,H)$. Then $U|_{T\c}$ is the identity and, w.r.t.~an orthonormal basis of $T$, we have $U|_T = \begin{smm} \alpha & -\beta \\ \beta & \alpha \end{smm}$, where $\alpha, \beta \in K$ are such that $\alpha^2 + \beta^2 = 1$. If $\beta = 0$, then $\alpha = 1$ or $\alpha = -1$ and hence $U|_T$ induces the identity on $P(T)$. If $\beta \neq 0$, $U|_T$ does not possess any eigenvector and hence $U|_T$ induces on $P(T)$ a map without fixpoints.

Conversely, let $\phi$ be an automorphism of $P(H)$ such that $\phi|_{L\c}$ is the identity and $\phi|_L$ is either the identity or does not possess any fixpoint. By Theorem \ref{thm:Wigner-for-OS}, $\phi$ is induced by an orthogonal operator $U$ such that $U|_{T\c}$ is the identity. W.r.t.~an orthonormal basis of $T$, $U|_T$ is of the form $\begin{smm} \alpha & -\beta \\ \beta & \alpha \end{smm}$ or $\begin{smm} \alpha & \beta \\ \beta & -\alpha \end{smm}$, where $\alpha^2 + \beta^2 = 1$. In the latter case, $U|_T$ has the distinct eigenvalues $1$ and $-1$, hence $\phi|_L$ has exactly two fixpoints. We conclude that $U|_T$ is of the form of the first matrix and hence $U \in \SO(T,H)$. By Proposition \ref{prop:circulation-group-1}, $\phi = \phi_U$ belongs to $C$.
\end{proof}

It seems finally natural to ask how $\Circ{P(H),\perp}$ is related to $\SO(H)$. By Proposition \ref{prop:circulation-group-1}, we know that $\Circ{P(H),\perp} \subseteq \{ \phi_U \colon U \in \SO(H) \}$: any circulation is induced by a rotation. Under an additional assumption, we can make a more precise statement. We call a field {\it Pythagorean} if any sum of two squares is itself a square.

In what follows, $\PSO(H) = \SO(H)/(\{I, -I\} \cap \SO(H))$ is the projective special orthogonal group of $H$.

\begin{proposition} \label{prop:circulation-group-2}
Assume that $K$ is Pythagorean. Then we have $\Circ{P(H),\perp} = \{ \phi_U \colon U \in \SO(H) \}$. Furthermore, the map $\SO(H) \to \Circ{P(H),\perp} \komma U \mapsto \phi_U$ is a surjective homomorphism. Its kernel is $\{I, -I\} \cap \SO(H)$, hence it induces an isomorphism between $\PSO(H)$ and $\Circ{P(H),\perp}$.
\end{proposition}

\begin{proof}
By Theorem \ref{thm:Wigner-for-OS}, $\SO(H) \to \Aut{P(H),\perp} \komma U \mapsto \phi_U$ is a homomorphism, whose kernel is $\{I, -I\} \cap \SO(H)$. By Proposition \ref{prop:circulation-group-1}, the images of the subgroups $\SO(T,H)$ of $\SO(H)$, where $T$ are the $2$-dimensional subspaces, under this homomorphism are exactly the basic circulation groups.

We shall show that $\SO(H)$ is generated by the basic rotations. Since $\Circ{P(H), \perp}$ is by definition generated by the basic circulations, the assertions will then follow.

Note first that, for any elements $\gamma, \delta \in K$ that are not both $0$, there are $\alpha, \beta, \rho \in K$ such that $\alpha^2 + \beta^2 = 1$, $\;\rho \neq 0$, and
\[ \begin{smm} \alpha & -\beta \\ \beta & \alpha \end{smm} \;
\begin{smm} \gamma \\ \delta \end{smm} \;=\; \begin{smm} \rho \\ 0 \end{smm}. \]
Indeed, let $\rho^2 = \gamma^2 + \delta^2$, $\;\alpha = \frac{\gamma}{\rho}$, and $\beta = -\frac{\delta}{\rho}$.

It follows that any matrix in $K^{n \times n}$ can be transformed by left multiplication with Givens rotations into row echelon form. When doing so with a matrix representing a rotation, the resulting matrix must be diagonal, an even number of the diagonal entries being $-1$ and remaining ones being $1$. We conclude that each rotation is the product of basic rotations in $2$-dimensional subspaces spanned by the elements of any given basis.
\end{proof}

\section{Embedding into $\Reals^n$}
\label{sec:spaces-over-subfields-of-R}

Our final aim is to present a condition with the effect that our orthogonality space arises from a quadratic space over an Archimedean field. In order to exclude the existence of non-zero infinitesimal elements, we shall require that our orthogonality space is, in a certain sense, simple.

An equivalence relation $\congr$ on an orthogonality space $(X,\perp)$ is called a {\it congruence} if any two orthogonal elements belong to distinct $\congr$-classes. Obviously, $X$ possesses at least one congruence, the identity relation, which we call {\it trivial}. For a congruence $\congr$ on $X$, we can make $X/{\congr}$ into an orthogonality space, called the {\it quotient orthogonality space}: for $e, f \in X$, we let $e/{\congr} \perp f/{\congr}$ if there are $e' \congr e$ and $f' \congr f$ such that $e' \perp f'$.

Given an automorphism $\phi$ of $(X,\perp)$, we call a congruence $\congr$ {\it $\phi$-invariant} if, for any $e, f \in X$, we have that $e \congr f$ is equivalent to $\phi(e) \congr \phi(f)$. If $\congr$ is $\phi$-invariant for every member $\phi$ of a subgroup $G$ of $\Aut{X,\perp}$, we say that $\congr$ is {\it $G$-invariant}.

We consider the following condition on $(X,\perp)$:

\begin{itemize}

\item[(R$_2$)] $(X, \perp)$ does not possess a non-trivial $\Circ{X,\perp}$-invariant congruence.

\end{itemize}

\begin{example}
Let again $\Reals^n$, $n \geq 1$, be endowed with the usual inner product. By Proposition \ref{prop:circulation-group-2}, $\Circ{P(\Reals^n),\perp}$ consists exactly of those automorphisms of $(P(\Reals^n),\perp)$ that are induced by some $U \in \SO(n)$. Moreover, $\SO(n)$ acts primitively on $P(\Reals^n)$, that is, no non-trivial partition of $P(\Reals^n)$ is invariant under $\SO(n)$. This means that no non-trivial partition of $P(\Reals^n)$ is invariant under $\Circ{P(\Reals^n)}$. In particular, the only $\Circ{P(\Reals^n)}$-invariant congruence is the identity relation. We conclude that $(P(\Reals^n), \perp)$ fulfils {\rm (R$_2$)}.
\end{example}

Let $H$ be a positive definite quadratic space over the ordered field $K$ as in Section \ref{sec:circulation-group}, that is, we assume that $H$ is of finite dimension $\geq 4$, each one-dimensional subspace of $H$ contains a unit vector, and $(P(H), \perp)$ fulfils (R$_1$).

Following Holland \cite{Hol2}, we define
\begin{align*}
I_K \;=\; & \{ \alpha \in K \colon \abs{\alpha} < \tfrac 1 n \text{ for all $n \in \Naturals \setminus \{0\}$} \}, \\
M_K \;=\; & \{ \alpha \in K \colon \tfrac 1 n < \abs{\alpha} < n \text{ for some $n \in \Naturals \setminus \{0\} $} \}
\end{align*}
to be the sets of {\it infinitesimal} and {\it medial} elements of $K$, respectively. Then $I_K$ is an additive subgroup of $K$ closed under multiplication; $M_K$ is a multiplicative subgroup of $K\withoutzero$; and we have $I_K \cdot M_K = I_K$ and $M_K + I_K = M_K$.

We call $K$ {\it Archimedean} if the only infinitesimal element is $0$. We have that $K$ is Archimedean exactly if all non-zero elements are medial. The following result is due to Holland \cite{Hol1}.

\begin{theorem} \label{thm:Archimedean-ordered-fields}
An Archimedean ordered field is order-isomorphic to an ordered subfield of $\Reals$.
\end{theorem}

We shall show that condition (R$_2$) implies $K$ to be Archimedean. Following again \cite{Hol2}, we define
\begin{align*}
I_H \;=\; & \{ x \in H \colon \herm x x \in I_K \}, \\
M_H \;=\; & \{ x \in H \colon \herm x x \in M_K \}
\end{align*}
to be the set of {\it infinitesimal} and {\it medial} vectors, respectively. Then $I_H$ is a subgroup of $H$ and we have $I_K \cdot M_H = M_K \cdot I_H = I_H$, $\;I_K \cdot I_H \subseteq I_H$, and $M_K \cdot M_H = M_H$. Furthermore, the Schwarz inequality implies that $\herm x y \in I_K$ if $x, y \in I_H \cup M_H$ and at least one of $x$ and $y$ is infinitesimal.

Furthermore, for $\lin x, \lin y \in P(H)$, we put $\lin x \similar \lin y$ if there are medial vectors $x' \in \lin x$ and $y' \in \lin y$ such that $x' - y' \in I_H$.

\begin{lemma} \label{lem:K-is-Archimedean}
Assume that $(P(H), \perp)$ fulfils {\rm (R$_2$)}. Then $K$ is an ordered subfield of the ordered field $\Reals$.
\end{lemma}

\begin{proof}
We first show that $\similar$ is an equivalence relation on $P(H)$. Clearly, $\similar$ is reflexive and symmetric. Let $x, y, z \in H$ be such that $\lin x \similar \lin y$ and $\lin y \similar \lin z$. Then there are $x' \in \lin x \cap M_H$, $\,y',y'' \in \lin y \cap M_H$, and $z' \in \lin z \cap M_H$ such that $x'-y', y''-z'' \in I_H$. Let $\alpha \in K$ be such that $y' = \alpha y''$. Then $\alpha^2 = \herm{y''}{y''}^{-1} \herm{y'}{y'} \in M_K$ and consequently $\alpha \in M_K$. Hence $\alpha z''$ is a medial vector as well, and $x' - \alpha z'' = (x'-y') + (y' - \alpha z'') = (x'-y') + \alpha (y''-z'') \in I_H$.

We claim that $\similar$ is a congruence. Let $x, y \in H\withoutzero$ be such that $\lin x \similar \lin y$. Then there are $x' \in \lin x \cap M_H$ and $y' \in \lin y \cap M_H$ such that $y' - x' \in I_H$. It follows $\herm{x'}{y'} = \herm{x'}{x' + (y'-x')} = \herm{x'}{x'} + \herm{x'}{y'-x'}$. Since $\herm{x'}{y'-x'} \in I_K$, we have $\herm{x'}{y'} \in M_K$. We have shown that $\lin x \not\perp \lin y$, because otherwise $\herm{x'}{y'} = 0$.

Let $\phi \in \Circ{P(H),\perp}$. Then $\phi$ is induced by an orthogonal operator $U$. For any $x, y \in H\withoutzero$, we have that $\lin x \similar \lin y$ implies $\lin{U(x)} \similar \lin{U(y)}$. Indeed, if $x' \in \lin x \cap M_H$ and $y' \in \lin y \cap M_H$ are such that $x' - y' \in I_H$, then also $U(x') \in \lin{U(x)} \cap M_H$ and $U(y') \in \lin{U(y)} \cap M_H$ are such that $U(x') - U(y') = U(x'-y') \in I_H$. We conclude that $\similar$ is $\phi$-invariant.

We have thus shown that $\similar$ is a $\Circ{P(H),\perp}$-invariant congruence on $P(H)$. By condition (R$_2$), $\similar$ is trivial.

Assume finally that $K$ contains the non-zero infinitesimal element $\delta$. For orthogonal unit vectors $u$ and $v$, we then have $\lin u \similar \lin{u+\delta v}$, because $u$ and $u+\delta v$ are medial vectors whose difference is infinitesimal. It follows that $\similar$ is non-trivial, a contradiction. We conclude that $K$ must be Archimedean.
\end{proof}

Again, we summarise our results.

\begin{theorem} \label{thm:main-representation}
Let $(X, \perp)$ a linear orthogonality space of finite rank $\geq 4$ that fulfils {\rm (R$_1$)} and {\rm (R$_2$)}. Then there is an ordered subfield $K$ of $\Reals$ and a positive-definite quadratic space $H$ over $K$, possessing unit vectors in each one-dimensional subspace, such that $(X, \perp)$ is isomorphic to $(P(H), \perp)$.
\end{theorem}

\section{Conclusion}
\label{sec:conclusion}

Being based on a binary relation about which not more than symmetry and irreflexivity is assumed, an orthogonality space is based on the sole idea of distinguishability of some abstract entities. We have seen that rather simple conditions lead us to the realm of inner-product spaces. We have made a further hypothesis according to which an orthogonality space possesses enough symmetries to allow, intuitively, a gradual transition from one entity to any other one, in a way that unconcerned elements are kept fixed. We were led to a linear structure not too far from a real Hilbert space -- a positive-definite quadratic space over an ordered field. An additional condition had the effect that the ordered field is a subfield of~$\Reals$.

Improvements of this work are certainly possible in a number of respects. First to mention, it would be interesting to clarify to which extent the idea of postulating gradual transitions between any two elements alone allows a reasonable structure theory for orthogonality spaces. An attempt in this direction is contained in our note \cite{Vet3}, where, however, the concrete formulation of the central condition has led to technical subtleties.

Furthermore, our guiding example has in this work not been the standard model of quantum mechanics but rather a (finite-dimensional) real Hilbert space. A complex linear space can be understood as a real linear space endowed with a complex structure. We should hence ask whether an analogous extension could be defined for the orthogonality spaces that we have considered here, such that we are led to an inner-product space over a subfield of the field that in quantum physics actually matters.

\vspace{2ex}

{\bf Acknowledgement.} The author acknowledges the support by the Austrian
Science Fund (FWF): project I 4579-N and the Czech Science Foundation
(GA\v CR): project 20-09869L.


\begin{thebibliography}{DuPr2}

\bibitem[Bru]{Bru} O. Brunet,
Orthogonality and dimensionality,
{\sl Axioms} {\bf 2} (2013), 477--489.

\bibitem[Dac]{Dac} J. R. Dacey,
``Orthomodular spaces'',
Ph.D.~Thesis, University of Massachusetts, Amherst 1968.

\bibitem[EGL1]{EGL1} K.~{Engesser}, D.~M. {Gabbay}, D.~{Lehmann} (Eds.),
``Handbook of quantum logic and quantum structure. Quantum structures'',
Elsevier, Amsterdam 2007.

\bibitem[EGL2]{EGL2} K.~{Engesser}, D.~M. {Gabbay}, D.~{Lehmann} (Eds.),
``Handbook of quantum logic and quantum structures. Quantum logic'',
Elsevier, Amsterdam 2009.

\bibitem[Ern]{Ern} M. Ern\'e,
Closure,
in: F. Mynard, E. Pearl (Eds.),
``Beyond topology'', Contemporary Mathematics 486,
American Mathematical Society, Providence 2009;
163--238.

\bibitem[Fin]{Fin} P. D. Finch,
Orthogonality relations and orthomodularity,
{\sl Bull. Aust. Math. Soc.} {\bf 2} (1970), 125--128.

\bibitem[Gro]{Gro} H. Gross,
Quadratic forms in infinite dimensional vector spaces'',
Birkhäuser, Boston-Basel-Stuttgart 1979.

\bibitem[Har]{Har} L. Hardy,
Quantum theory from five reasonable axioms,
arXiv:quant-ph/0101012.

\bibitem[HePu]{HePu} J. Hedl\' ikov\' a, S. Pulmannov\' a,
Orthogonality spaces and atomistic orthocomplemented lattices,
{\sl Czech. Math. J.} {\bf 41} (1991), 8--23.

\bibitem[Hol1]{Hol1} S. S. Holland,
Orderings and square roots in $\star$-fields,
{\sl J. Algebra} {\bf 46} (1977), 207 - 219.

\bibitem[Hol2]{Hol2} S. S. Holland,
$\star$-valuations and ordered $\star$-fields,
{\sl Trans. Am. Math. Soc.} {\bf 262} (1980), 219 - 243.

\bibitem[Hol3]{Hol3} S. S. Holland,
Orthomodularity in infinite dimensions; a theorem of M. Sol\` er,
{\it Bull. Am. Math. Soc., New Ser.} {\bf 32} (1995), 205--234.

\bibitem[Mac]{Mac} M. D. MacLaren,
Atomic orthocomplemented lattices,
{\sl Pac. J. Math.} {\bf 14} (1964), 597--612.

\bibitem[MaMa]{MaMa} F. Maeda, S. Maeda,
``Theory of symmetric lattices'',
Springer-Verlag, Berlin~-~Heidelberg~-~New York 1970.

\bibitem[May]{May} R. Mayet,
Some characterizations of the underlying division ring of a Hilbert lattice by automorphisms,
{\sl Int. J. Theor. Phys.} {\bf 37} (1998), 109--114.

\bibitem[Pir]{Pir} C. Piron,
``Foundations of quantum physics'', W.A. Benjamin, Reading 1976.

\bibitem[Piz]{Piz} R. Piziak,
Orthomodular lattices and quadratic spaces: a survey,
{\sl Rocky Mt. J. Math.} {\bf 21} (1991), 951 - 992.

\bibitem[Pre]{Pre} A.~Prestel,
``Lectures on formally real fields'', Springer-Verlag, Berlin etc.~1984.

\bibitem[Pul]{Pul} S. Pulmannov\' a,
Representations of quantum logics and transition probability spaces,
in:
E. I. Bitsakis, E. Nicolaidis (Eds.),
``The concept of probability'', Proceedings of the Delphi Conference (Delphi 1987),
Springer-Verlag 1989,
pp. 51--59.

\bibitem[Rod]{Rod} M. S. Roddy,
An orthomodular lattice,
{\sl Algebra Univers.} {\bf 29} (1992), 564--579.

\bibitem[Sch]{Sch} W. Scharlau,
``Quadratic and Hermitian forms'',
Springer-Verlag, Berlin 1985.

\bibitem[Uhl]{Uhl} U. Uhlhorn,
Representation of symmetry transformations in quantum mechanics,
{\sl Ark. Fys.} {\bf 23} (1963), 307--340.

\bibitem[Vet1]{Vet1} Th. Vetterlein,
Orthogonality spaces of finite rank and the complex Hilbert spaces,
{\sl Int. J. Geom. Methods Mod. Phys.} {\bf 16} (2019), 1950080.

\bibitem[Vet2]{Vet2} Th. Vetterlein,
Orthogonality spaces arising from infinite-dimensional complex Hilbert spaces,
{\it Int. J. Theor. Phys.}, to appear.

\bibitem[Vet3]{Vet3} Th. Vetterlein,
Orthogonality spaces allowing gradual transitions,
Proceedings of the 11th conference of the European Society of Fuzzy Logic and Technology (Prague 2019),
Atlantis Press 2019;
192--197.

\bibitem[Wlc]{Wlc} A. Wilce,
Test spaces,
in: \cite{EGL2}, pp.~443--549.

\end{thebibliography}
\end{document}